\documentclass[authoryear,preprint,12pt]{elsarticle}
\usepackage{amssymb}

\usepackage{amsmath,amsthm,amssymb}
\usepackage{graphicx}
\usepackage{setspace}

\journal{Journal of Theoretical Biology}
\newtheorem{thm}{Theorem}

\begin{document}

\begin{frontmatter}

\ead{zhangjiang@bnu.edu.cn}
\title{Scaling Behaviors of Weighted Food Webs as Energy Transportation Networks}


\author{Jiang Zhang, Liangpeng Guo}

\address{Department of Systems Science, School of Management, Beijing Normal University, Beijing, 100875}

\begin{abstract}
Food webs can be regarded as energy transporting networks in which
the weight of each edge denotes the energy flux between two species.
By investigating 21 empirical weighted food webs as energy flow
networks, we found several ubiquitous scaling behaviors. Two random
variables $A_i$ and $C_i$ defined for each vertex $i$, representing
the total flux (also called vertex intensity) and total indirect
effect or energy store of $i$, were found to follow power law
distributions with the exponents $\alpha\approx 1.32$ and
$\beta\approx 1.33$, respectively. Another scaling behavior is the
power law relationship, $C_i\sim A_i^\eta$, where $\eta\approx
1.02$. This is known as the allometric scaling power law
relationship because $A_i$ can be treated as metabolism and $C_i$ as
the body mass of the sub-network rooted from the vertex $i$,
according to the algorithm presented in this paper. Finally, a
simple relationship among these power law exponents,
$\eta=(\alpha-1)/(\beta-1)$, was mathematically derived and tested
by the empirical food webs.
\end{abstract}


\end{frontmatter}


\section{Introduction}
Scientists look for universal patterns of complex systems because
such invariant features may help to unveil the principles of system
organization\citep{Waldrop1992}. Complex network studies can not
only provide a unique viewpoint of nature and society but also
reveal ubiquitous patterns, e.g., small world and scale free,
characteristic of various complex
systems\citep{watts1998,AlbertBarabasi2002}. However, ecological
studies have shown that binary food webs, which depict trophic
interactions in ecosystems, refuse to become part of the small world
and scale free networks family \citep{Montoya2002,Dunne2002}.
Although some common features, including "two degrees separation",
which means the very small average distance, are shared among food
webs \citep{Williams2002}, other meaningful attributes such as
degree distribution and clustering coefficient change with the size
and complexity (connectance) of the network \citep{Dunne2002}.

Weight information of complex networks such as air traffic network
or metabolism networks, etc., can reveal more unique patterns and
features that are never found in binary
relationships\citep{Barrat2004,Almaas2004,Montis2007}. Food web
weights have two different, yet correlated, meanings in ecology. One
is the strength of the trophic
interaction\citep{Emmerson2004,Berlow2004}; the other is the amount
of energy flow. They are correlated because interaction strength is
the per capita measure of energy flow. Interaction strength-based
weighted food webs exhibit new features such as a relationship
between predator-prey body size ratio and interaction strength
\citep{McCann1998,Wootton2002,Berlow2004}.

Additionally, weights of food webs can also be denoted as the total
amount of energy flow between two species when the whole system is
in the steady state. Studies of energy flow networks in ecosystems
have a long history
\citep{odum1988,Finn1976,Szyrmer1987,Higashi1986,Baird1989,higashi1993,Patten1981,Patten1982}.
Many systematic indicators were designed to depict the macro-state
of energy flow in ecosystem \citep{Fath1999,Fath2001}, of which some
can not only reflect the direct energy flows between species but
also indirect effects and inter-dependence of
species\citep{Fath1999,Finn1976}. Although some important
discoveries were made about the general structure and function of
ecological
networks\citep{Fath1998,Hannon1973,levine1980,Hannon1986,Patten1985,Ulanowicz1986,Ulanowicz1997b},
few focus on power law distributions and
relations\citep{Ulanowicz1990}.

Allometric scaling is an important universal pattern of flow
systems. \cite{kleiber1932} found that the metabolism and body size
of all species follow a ubiquitous power law relationship, with an
exponent around 3/4. \cite{west1997} and \cite{banavar1999}
explained this pattern as an emergent property of nutrient and
energy transportation networks. This recognition encouraged people
to realize that allometric scaling may be a universal feature for
all transportation systems. \cite{garlaschelli2003} extended
Banavar's approach to binary food webs and found a similar
allometric scaling power law relationship. Although Garlaschelli's
method as an algorithm had been applied to various networks,
including the worldwide trade network\citep{Duan2007} and tree of
life\citep{Herrada2008}, it had several shortcomings. The first step
of his algorithm is to obtain a spanning tree by cutting many edges
in the original network so that a certain amount of information is
lost\citep{garlaschelli2003}. \cite{allesina2005} improved this
method by reducing the original network to a directed acyclic graph.
Although less information is lost, cutting edges is still
unavoidable. The second shortcoming of Garlaschelli's approach and
Allesina's improvement is they can be applied to binary networks,
but not weighted ones.

This paper will combine the successful approaches in complex
weighted networks and earlier studies on ecological flow networks to
reveal the underlying heterogeneities and universal scaling
behaviors of food webs. The study is organized as follows. In
section \ref{sec.methods}, the basic ideas and steps for obtaining
$A_i$ and $C_i$ are introduced. Afterwards, we apply these tools on
21 empirical food webs with energy flow information. Section
\ref{sec.powerlawdisai} and \ref{sec.powerlawdisci} study the power
law distributions of $A_i$ and $C_i$. We extend Garlarschelli's
approach to weighted food webs without cutting edges. The allometric
scaling power law relationship between $A_i$ and $C_i$ is shown in
section \ref{sec.allometricscaling}. A simple mathematical
relationship among scaling exponents of power law distributions and
power law relation is derived and tested using empirical food webs
in section \ref{sec.exponentrelation}. Finally, the ecological
meaning of $A_i$ and $C_i$, the distributions of flux matrix and
fundamental matrix, consideration of node information, etc., are
discussed in section \ref{sec.discussion}. A simple example of our
approach, a comparison to the existing methods, and the theorem
regarding power law exponents are presented in the Appendix.

\section{Methods}
\label{sec.methods}

In this section, we outline the basic idea and mathematical
definition of our method. One simple example showing how the
approach works will be discussed in the \ref{sec.example} in detail.

\label{sec.representation}
\subsection{Flux Matrix}
An ecological energy flow network is a weighted directed graph that
represents relationships of ecological energy transfer. For a given
graph, a matrix called flux matrix in this paper can be defined as
representing the energy flux between species.
\begin{equation}
\label{eqn.Fluxmatrix}
 F_{(N+2)\times (N+2)}=\{f_{ij}\}_{(N+2)\times (N+2)}
\end{equation}
where $f_{ij}$ is the energy flux from species $i$ to $j$. Two
special vertices represent the environment: vertex $0$ and vertex
$N+1$. Vertex $0$ denotes the source of energy flow, whereas vertex
$N+1$ represents the sink. We expect that the dissipative and
exported energy will flow to vertex $N+1$. Therefore, there are in
total $(N+2)\times (N+2)$ entries in the flux matrix.

\subsection{Fundamental Matrix}

Suppose that the flow network is balanced, meaning that the total
influx equals the efflux for each vertex $i\in [1,N]$. We can then
define an $N\times N$ matrix $M$ from $F$ follows,

\begin{equation}
\label{defm} m_{ij}={f_{ij} /(\sum_{k=1}^{N+1}{f_{ik}})}, \forall
i,j \in [1,N]
\end{equation}

and the fundamental matrix can be derived as
\begin{equation}
\label{eqn.utilizationmatrix}U = I+M+M^2+ \cdots =
\sum_{i=0}^{\infty}{M^i}=(I-M)^{-1}
\end{equation}

where, $I$ is the unity matrix. Any element $u_{ij}$ in $U$ matrix
denotes the influence $i$ to $j$ along all possible pathways. $U$
matrix was first introduced in economic input-output
analyses\citep{leontief1951,Leontief1966} to indicate the direct and
indirect effects of good flows in various economic sectors.
\cite{Hannon1973} was the first to apply this matrix to
ecology\citep{Fath1999,ulanowicz2004}.

Given the flux matrix and fundamental matrix, two vertex-related
variables, $A_i$ and $C_i$, which will later be shown to follow
power law distributions, are defined.

\subsection{$A_i$}

We can calculate the total flux through any given vertex $i$
according to $F$. This value is also called node intensity in
complex weighted network studies\citep{Almaas2004}. Because the
network is balanced, we need only calculate the efflux of each node
as $A_i$,
\begin{equation}
\label{eqn.ai} A_i=\sum_{j=1}^{N+1}{f_{ij}}, \forall i \in [1,N]
\end{equation}

\subsection{$C_i$}

Another vertex-related index called $C_i$ can be defined to reflect
the total indirect effects or the total energy store of the
sub-network rooted from vertex $i$.

\begin{equation} \label{eqn.ci}
C_i=\sum_{k=1}^N{\sum_{j=1}^N{(f_{0j}{u_{ji}}/ u_{ii})u_{ik}}}
\end{equation}

We will provide an explanation of the indicator $C_i$ in
\ref{sec.example} by a simple example.

$A_i$ is the total flow-through or intensity of vertex $i$. $C_i$ is
the total influence of vertex $i$ on all vertices in the whole
network. Suppose that of the many particles flowing in the network
\citep{higashi1993}, those passing vertex $i$ will be colored red.
$C_i$ would then be the total number of red particles flowing in the
network. Actually, these two variables are extended from the
approach of \cite{garlaschelli2003} to calculate the allometric
scaling of food webs (see \ref{sec.comparison}).

\subsection{Balancing the Network}

Sometimes the empirical network is not strictly balanced. To
facilitate our algorithm, we can balance them artificially. Suppose
$\sum_{j=0}^{N}{f_{ji}}\neq \sum_{j=1}^{N+1}{f_{ij}}$ for vertex
$i$. We can add an edge with the weight $|f'_{ij}|$,
$f'_{ij}=\sum_{j=0}^{N}{f_{ji}}-\sum_{j=1}^{N+1}{f_{ij}}$ to connect
the vertex $i$ to $N+1$ or $0$. If $f'_{ij}>0$, the direction of
this artificial edge is from $i$ to $N+1$. If $f'_{ij}<0$, the
direction is from $0$ to $i$. Normally, the artificial edges have
very small weights because most empirical food webs are almost
balanced already.

\subsection{Power Laws}

After calculating the indicators of $A_i$ and $C_i$ for each vertex
$i$, we will show that they follow the power law distributions in
the high tails, which means that,

\begin{equation}
\label{eqn.powerlawai} P(A_i>x)\thicksim x^{1-\alpha}
\end{equation}

and

\begin{equation}
\label{eqn.powerlawci} P(C_i>y)\thicksim y^{1-\beta}
\end{equation}

for given $A_i$ and $C_i$ which are larger than given thresholds
$x_0,y_0$ \citep{Clauset2007}, and where $\thicksim $ represents
"proportional to." The cumulative probability distribution curves
will be shown and the power law exponents $\alpha,\beta$ calculated
in the next section.

Furthermore, we will show that $A_i$ and $C_i$ satisfy a power law
relationship,

\begin{equation}
\label{eqn.allometricscaling} C_i\thicksim A_i^{\eta}
\end{equation}

This relationship is also called the allometric scaling law because
$A_i$ represents metabolism and $C_i$ is the equivalent body mass of
the sub-system rooted from vertex $i$ (see \ref{sec.example}).

\section{Results}
\subsection{Dataset}
Twenty one food webs containing energy flow information from
different habitats were studied(Table \ref{tab.foodwebdata}). These
food webs were obtained from an online database\footnotemark
\footnotetext[1]{http://vlado.fmf.uni-lj.si/pub/networks/data/bio/foodweb/foodweb.htm},
and most are from published papers. In Table \ref{tab.foodwebdata},
we list the name and the number of nodes ($|N|$) and edges ($|E|$)
in each web. The number of nodes does not include the "respiration"
node, and the number of edges only counts the energy flows between
species, and does not include the edges from (to) "input" and
"output." The weights of edges in these food webs are energy flows
whose values vary across a large range because the units and time
scales of the measurements are very different.

\begin{table}
\centering
 \caption{Empirical Food Webs and Their Topological
 Properties}{\small ($|N|$ stands for the number of vertices of the network and $|E|$ is the number of edges. The webs are sorted by their number of edges.)}
 \label{tab.foodwebdata}
\begin{tabular}{llllll}
\hline Food web & Abbre.  & $|N|$ & $|E|$\\
\hline

Florida Bay, Dry Season & BayDry  & 127 & 1969\\
Florida Bay, Wet Season & BayWet  & 127 & 1938\\
Florida Bay & Florida & 127 & 1938\\
Mangrove Estuary, Dry Season & MangDry  & 96 & 1339\\
Everglades Graminoid Marshes  & Everglades  & 68 & 793\\
Everglades Graminoids, Dry Season  & GramDry  & 68 & 793\\
Everglades Graminoids, Wet Season  & GramWet  & 68 & 793\\
Cypress, Dry Season & CypDry  & 70 & 554\\
Cypress, Wet Season & CypWet  &70 & 545\\
Mondego Estuary - Zostrea site & Mondego  & 45 & 348\\
St. Marks River (Florida) & StMarks  & 53 & 270\\
Lake Michigan & Michigan  & 38 & 172\\
Narragansett Bay & Narragan  & 34 & 158\\
Upper Chesapeake Bay in Summer &ChesUpper & 36 & 158\\
Middle Chesapeake Bay in Summer &ChesMiddle  & 36 & 149\\
Chesapeake Bay Mesohaline Net &Chesapeake  & 38 & 122\\
Lower Chesapeake Bay in Summer & ChesLower & 36 & 115\\
Crystal River Creek (Control) &CrystalC  & 23 & 81\\
Crystal River Creek (Delta Temp) &CrystalD  & 23 & 60\\
Charca de Maspalomas  & Maspalomas  & 23 & 55\\
Rhode River Watershed - Water Budget & Rhode  & 19 & 35\\
 \hline
\end{tabular}
\end{table}

After studying the scaling laws of these food webs, we divided the
results into two main sections. First, the power law distributions
that reflect the heterogeneity of $A_i$ and $C_i$ are shown. Second,
the allometric scaling power law relationship that depicts the
self-similar structures of energy flows is discussed.

\subsection{Power Law Distributions of $A_i$}
\label{sec.powerlawdisai}

We calculated the random variables $A_i$ for each of the 21
empirical food webs. Four of them are selected to plot in Figure
\ref{fig.aidistribution}.

\begin{figure}
\centerline {\includegraphics{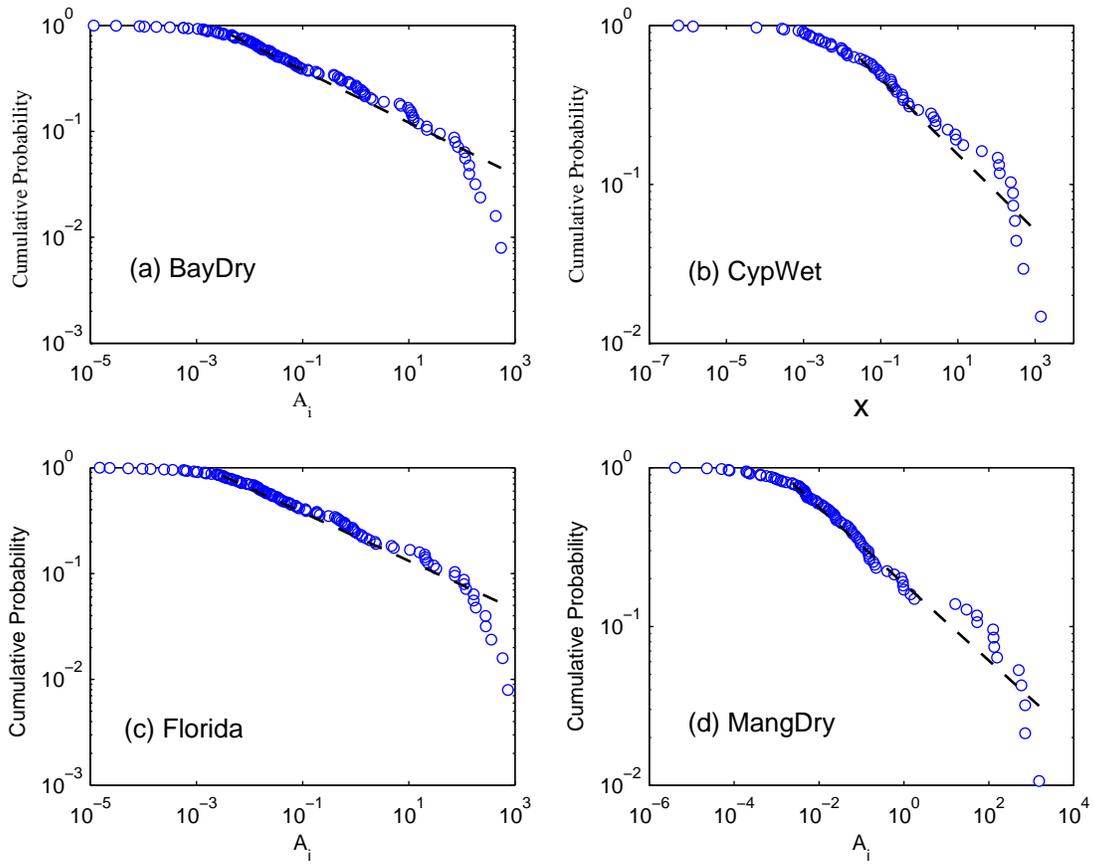}} \vskip3mm
\caption{Cumulative probability of the variable $A_i$ for four
selected food webs with the best fitted lines. The fitted lines
start from $x_0$, and their slopes are listed in Table
\ref{tab.aidistribution}}\label{fig.aidistribution}
\end{figure}

$A_i$s of real food webs decay as power law in the high tail;
therefore, the fitted lines start from given lower bounds $x_0$
\citep{Clauset2007}. These figures show that $A_i$ follows the power
law distribution.

There are obvious cutoffs in the tails that may be attributed to
sampling effects\citep{Newman2005}. According to equation
\ref{eqn.powerlawai}, the cumulative probability decays as $x$
slowly. The probability of finding a higher $A_i$ value (in the tail
of the curve) is very small. Therefore, the number of samples in
this interval becomes very few and statistical fluctuations are
unavoidably large as a fraction of sample number. This phenomenon is
obvious in other fields such as income\citep{Clementi2006}, personal
donations\citep{ChenWang2009} and the number of species per genus of
mammals\citep{Newman2005}.

The scaling exponent $\alpha$ for each food web was estimated
according to the maximum likelihood approach \citep{Clauset2007}.
The exponents and relative errors of power law fittings for all 21
food webs are listed in Table \ref{tab.aidistribution}.

\begin{table}
 \centering
 \caption{Parameters of the power law distributions for $A_i$}{\small
 ($\alpha$ is the power law distribution exponent of equation \ref{eqn.powerlawai};
 $x_0$ is the smallest value of $A_i$ that follows power law, $x_{max}$ is the largest $A_i$;
$D$ is the  KS statistic; $\sigma$ is the quantile of 95\%
confidence interval; and the number of samples is the total number
of nodes following the power law distribution. The webs are sorted
by their number of edges)}
 \label{tab.aidistribution}
\begin{tabular}{lccccc}
\hline Food web & $\alpha$ & $x_0/x_{max}$ & $D$ & $\sigma$ & No. of samples\\
\hline
Baydry&1.25&8.45e-006&0.09&0.13&102\\
Baywet&1.23&4.24e-006&0.08&0.13&105\\
Florida&1.23&4.24e-006&0.08&0.13&105\\
Mangdry&1.24&1.55e-006&0.06&0.16&75\\
Everglades&1.23&4.20e-006&0.12&0.20&44\\
Gramdry&1.25&9.55e-006&0.11&0.21&41\\
Gramwet&1.23&4.20e-006&0.12&0.20&44\\
Cypdry&1.28&5.95e-005&0.10&0.20&44\\
Cypwet&1.23&2.06e-005&0.10&0.21&42\\
Mondego&1.33&3.87e-004&0.11&0.26&26\\
StMarks&1.35&9.36e-004&0.17&0.21&43\\
Michigan&1.32&1.80e-003&0.18&0.31&18\\
Narragan&1.26&1.10e-004&0.16&0.23&33\\
ChesUpper&1.28&9.54e-004&0.21&0.23&32\\
ChesMiddle&1.29&9.47e-004&0.17&0.24&30\\
Chesapeake&1.41&8.25e-003&0.22&0.29&21\\
ChesLower&1.77&5.79e-002&0.20&0.33&16\\
CrystalC&1.27&1.04e-004&0.13&0.32&17\\
CrystalD&1.25&3.10e-005&0.12&0.30&19\\
Maspalomas&1.54&2.22e-002&0.18&0.29&20\\
Rhode&1.56&1.35e-002&0.17&0.34&15\\

 \hline
\end{tabular}
\end{table}

In Table \ref{tab.aidistribution}, $x_0$s represent the lower bounds
of the power law distributions. We normalized $x_0$ by dividing the
maximum $A_i$ of each food web to avoid the large range variance of
$x_0$ among different webs because their units and the measurement
time scales are very different.

$D$ is the statistic of the KS
test\citep{Rousseau2000,Goldstein2004}. Its value reflects the
maximum distance between the cumulative probability of real data and
the fitted model. Therefore, the smaller $D$ values indicate the
better power law fitting. $\sigma$ is the quantile of the 95\%
confidence interval for different numbers of
samples\citep{Noether1967}, and is only a reference for $D$. If $D$
is smaller than $\sigma$, then we should accept the power-law
hypothesis\citep{Noether1967}. From Table \ref{tab.aidistribution},
we know that all food webs pass the KS test. In the last column,
"No. of samples" means the number of sample points that are larger
than $x_0$ and follow the power law distribution.

By comparing different rows, we know that the food webs with more
edges can be better described by power laws because their $D$s are
smaller. Further, the scaling exponent $\alpha$ and $x_0/x_{min}$
increase as the scale of the network decreases. All $\alpha$ values
fall into the interval $[1.23,1.77]$, with an average of $1.32$.

The power law distributions of $A_i$s reflect the heterogeneities of
energy flux. Few nodes possess high $A_i$ values, while most nodes
only share a small fraction of the energy flux. The exponent of
power law reflects the degree of heterogeneity of the whole network.
Therefore, larger food webs are more heterogeneous than smaller ones
because their exponents are lower (Table \ref{tab.aidistribution}).
Although the power law distribution of $A_i$ can not give us
concrete information about each vertex
\citep{Fath1999,Patten1981,Patten1982}, it helps us to understand
the network as a whole.

\subsection{Power Law Distributions of $C_i$}
\label{sec.powerlawdisci}

The same approach can be applied to $C_i$s. The distributions of
$C_i$s for four selected food webs are shown in Figure
\ref{fig.cidistribution}.

\begin{figure}
\centerline {\includegraphics{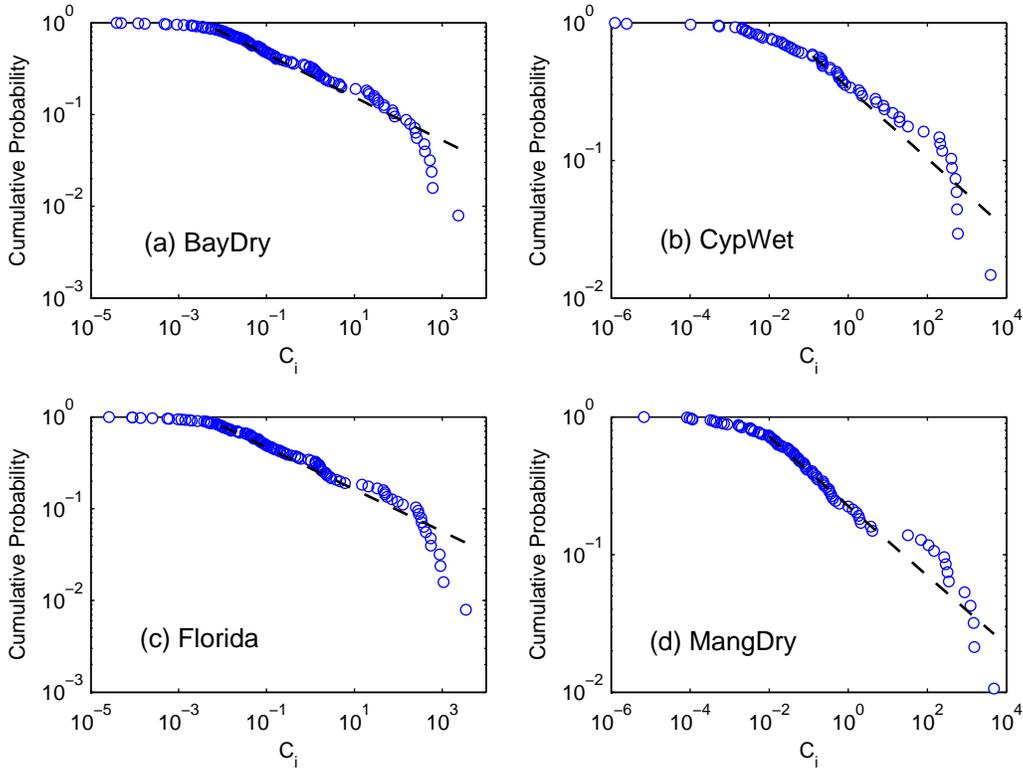}} \vskip3mm
\caption{Cumulative probability of the variable $C_i$s for four
selected food webs with the best fitted lines. The slopes of the
fitted lines are listed in Table
\ref{tab.cidistribution}}\label{fig.cidistribution}
\end{figure}

The curves of $C_i$ distributions are very similar to the curves in
Figure \ref{fig.aidistribution}. The estimated exponents and the KS
test parameters are listed in Table \ref{tab.cidistribution}.

\begin{table}
\centering
 \caption{Parameters of the power law distributions for $C_i$}
{\small
 ($\beta$ is the power law distribution exponent of equation \ref{eqn.powerlawci};
 $y_0$ is the smallest value of $C_i$ that follows the power law, $y_{max}$ is the largest $C_i$;
 $D$ is the KS statistic; $\sigma$ is the quantile of 95\%
confidence interval; the number of samples is the total number of
nodes following the power law distribution. The webs are sorted by
their number of edges)}
 \label{tab.cidistribution}
\begin{tabular}{lccccc}
\hline Food web & $\beta$ & $y_0/y_{max}$ & $D$ & $\sigma$ & No. of samples\\
\hline
Baydry&1.23&2.97e-006&0.08&0.13&107\\
Baywet&1.23&2.71e-006&0.08&0.14&101\\
Florida&1.23&2.71e-006&0.08&0.14&101\\
Mangdry&1.25&2.03e-006&0.07&0.16&68\\
Everglades&1.23&2.16e-006&0.10&0.20&44\\
Gramdry&1.25&4.57e-006&0.10&0.21&42\\
Gramwet&1.23&2.16e-006&0.10&0.20&44\\
Cypdry&1.28&4.88e-005&0.10&0.20&44\\
Cypwet&1.25&3.13e-005&0.11&0.21&39\\
Mondego&1.29&9.61e-005&0.12&0.24&30\\
StMarks&1.42&1.27e-003&0.16&0.21&40\\
Michigan&1.30&7.22e-004&0.17&0.31&18\\
Narragan&1.24&1.22e-004&0.17&0.23&33\\
ChesUpper&1.28&6.03e-004&0.20&0.23&32\\
ChesMiddle&1.26&3.88e-004&0.19&0.24&30\\
Chesapeake&1.61&2.41e-002&0.18&0.33&16\\
ChesLower&1.84&5.02e-002&0.22&0.35&14\\
CrystalC&1.28&1.04e-004&0.13&0.32&17\\
CrystalD&1.23&1.35e-005&0.12&0.29&20\\
Maspalomas&1.56&1.63e-002&0.19&0.29&20\\
Rhode&1.52&8.93e-003&0.23&0.32&17\\

 \hline
\end{tabular}
\end{table}

Comparing Table \ref{tab.cidistribution} with Table
\ref{tab.aidistribution}, the exponents of $C_i$ distributions are
slightly higher than those of $A_i$ distributions. The exponent
$\beta$ and KS test statistic $D$ decrease with the scale of the
network. The average exponent of these food webs is $1.33$, with all
$C_i$ values falling into the interval $[1.23, 1.84]$.

We have studied the heterogeneity of $C_i$s node by node. The nodes
with high $A_i$ values always have high $C_i$ values, indicating a
possible positive correlation between $A_i$ and $C_i$.

\subsection{Allometric Scaling Relations}
\label{sec.allometricscaling}

The similarity between Figure \ref{fig.cidistribution} and Figure
\ref{fig.aidistribution} shows that there must be some connections
between $A_i$ and $C_i$. Allometric scaling of these flow networks
revealed that the relationship between $A_i$ and $C_i$ is actually a
power law.

\begin{figure}
\centerline {\includegraphics{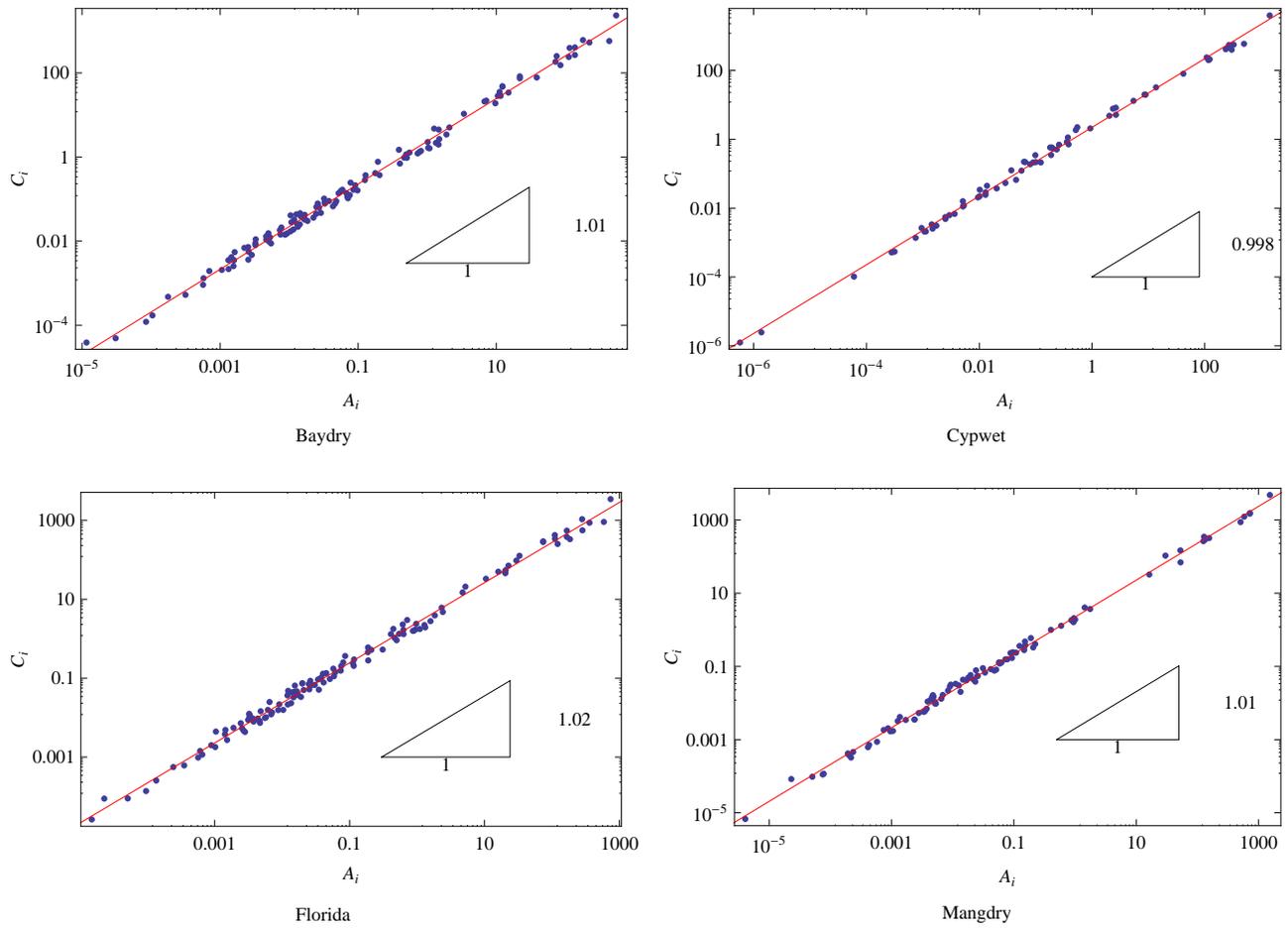}} \vskip3mm
\caption{Allometric Scaling relationship between $A_i$s and $C_i$s
of four selected food webs with the best fitted
lines}\label{fig.allometricscaling}
\end{figure}

As shown in Figure \ref{fig.allometricscaling}, the sample points
aggregate around their fitted lines very well. This relationship is
ubiquitous for all 21 food webs as shown in Table
\ref{tab.allometricscaling}.

\begin{table}
 \centering
 \caption{Allometric scaling of Empirical Food Webs}
 {\small(The second column lists $\eta$s of each food web with the errors; The webs are sorted by their number of edges)\\}
 \label{tab.allometricscaling}
\begin{tabular}{lccc}
\hline Food web & $\eta$ & $R^2$ \\ \hline
Baydry&1.01$\pm$0.01&0.9946\\
Baywet&1.02$\pm$0.01&0.9946\\
Florida&1.02$\pm$0.01&0.9946\\
Mangdry&1.01$\pm$0.01&0.9967\\
Everglades&1.02$\pm$0.01&0.9992\\
Gramdry&1.03$\pm$0.01&0.9990\\
Gramwet&1.02$\pm$0.01&0.9992\\
Cypdry&1.00$\pm$0.02&0.9957\\
Cypwet&1.00$\pm$0.01&0.9970\\
Mondego&1.01$\pm$0.01&0.9989\\
StMarks&1.03$\pm$0.04&0.9784\\
Michigan&1.01$\pm$0.01&0.9986\\
Narragan&1.01$\pm$0.04&0.9910\\
ChesUpper&1.05$\pm$0.02&0.9966\\
ChesMiddle&1.04$\pm$0.02&0.9959\\
Chesapeake&0.99$\pm$0.02&0.9966\\
ChesLower&1.05$\pm$0.02&0.9974\\
CrystalC&0.89$\pm$0.23&0.7706\\
CrystalD&0.90$\pm$0.23&0.7722\\
Maspalomas&0.96$\pm$0.09&0.9656\\
Rhode&0.83$\pm$0.17&0.8658\\

 \hline
\end{tabular}
\end{table}

We used the minimum square error method to find the best-fitted line
(Table \ref{tab.allometricscaling}). $R^2$s were larger than $0.9$
for all food webs except CrystalC, CrystalD and Rhode, whose scales
are very small ($|N|<23$). The $R^2$s and exponents decrease with
the scale of the network because the statistical significance
decreases as the number of samples declines. All exponents $\eta$
fall into the interval $[0.83,1.05]$. The mean value of $\eta$s for
these food webs, except CrystalC,CrystalD and Rhode, is $1.02$.

We also show the $A_i$ and $C_i$ values of root nodes for all food
webs, and fit them with a line on the log-log plot (Figure
\ref{fig.rootpower}). These power law relations reflect the
self-similar nature of the weighted food webs.

\begin{figure}
\centerline {\includegraphics[width=12cm]{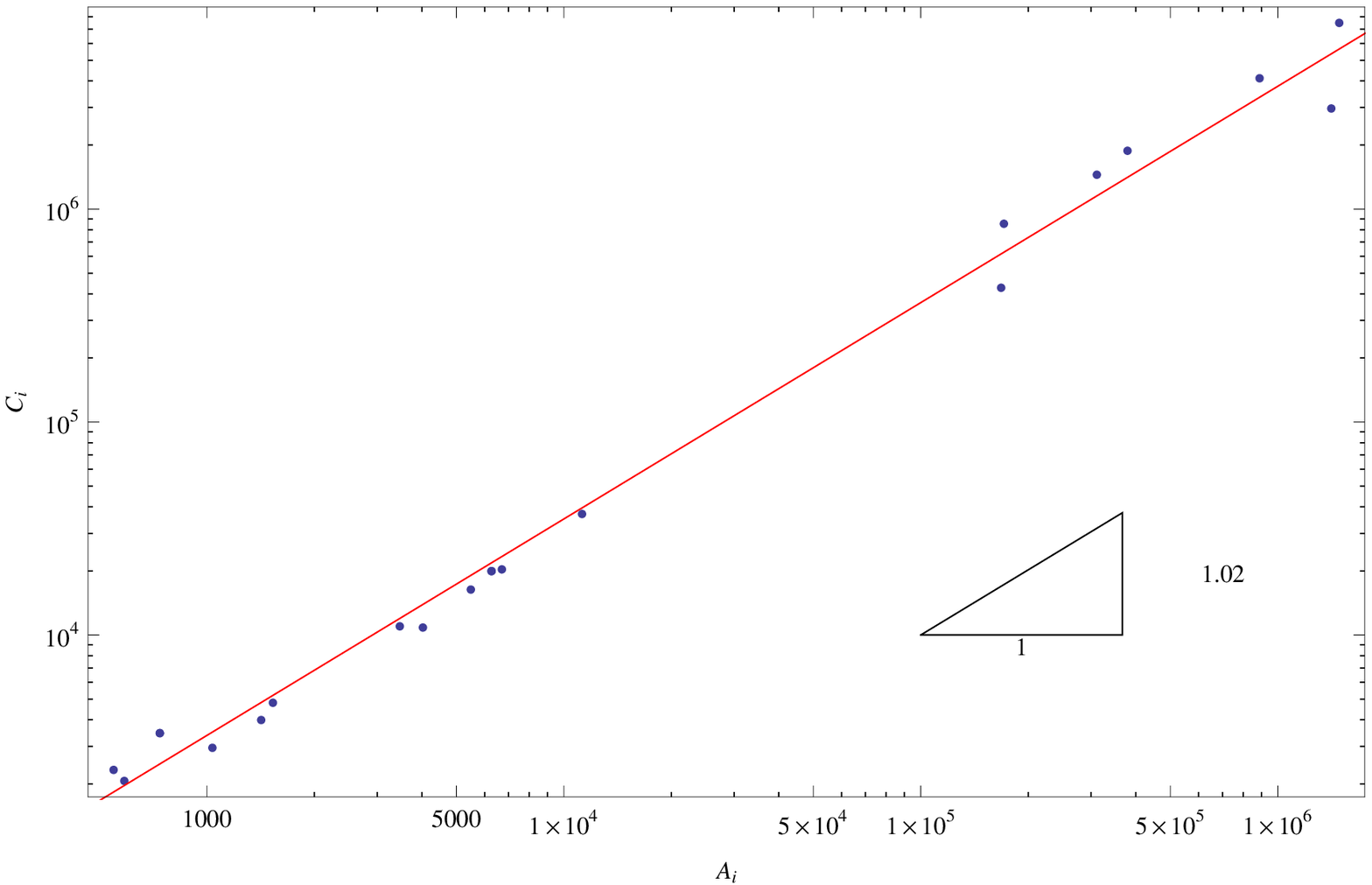}} \vskip3mm
\caption {$A_i$, $C_i$ plot for root nodes of 21 food webs. The
slope $\eta$ (the power law exponent) is
1.02$\pm$0.02}\label{fig.rootpower}
\end{figure}

\subsection{Relationship of Scaling Exponents}
\label{sec.exponentrelation}

As we have shown, $A_i$ and $C_i$ are random variables following
power law distributions with scaling exponents $\alpha$ and $\beta$,
respectively. They also follow a power law relationship with the
scaling exponent $\eta$. Is there any universal relationship among
$\alpha,\beta$ and $\eta$?

Actually, for any random variables following power law distributions
and power law relations, we can prove a mathematical theorem (see
\ref{sec.theorem}). According to this theorem, the power law
exponents $\alpha,\beta,\eta$ of the food webs should also satisfy
equation \ref{eqn.exponentsrelation}. We tested this hypothesis by
calculating $(\alpha-1)/(\beta-1)-\eta$ for each of the 21 empirical
food webs to obtain Figure \ref{fig.exponents}.

\begin{figure}
\centerline {\includegraphics[width=12cm]{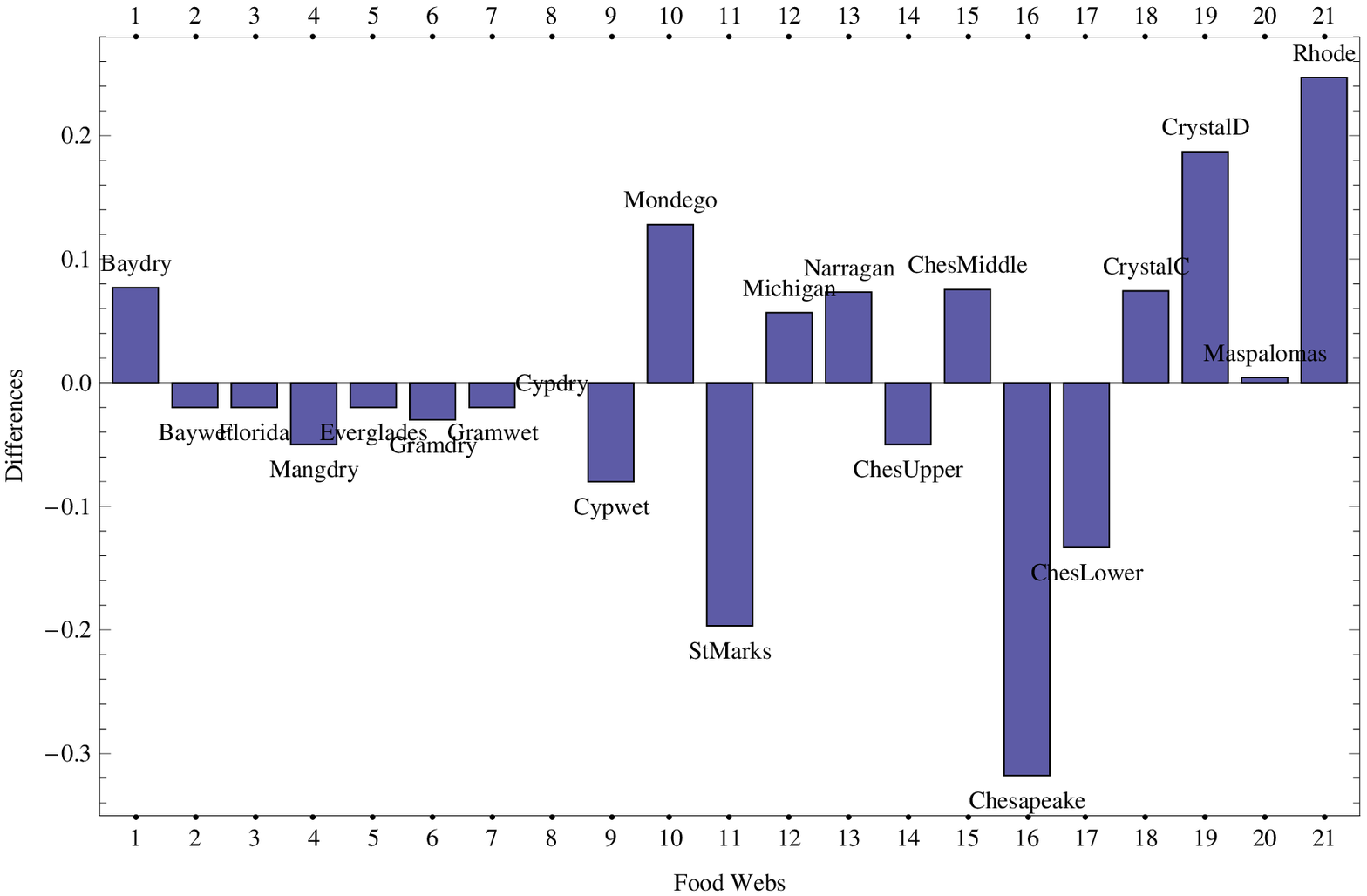}} \vskip3mm
\caption {The differences between $\eta$ and
${(\alpha-1)}/{(\beta-1)}$ are calculated for all food webs to test
the exponents relationship (equation \ref{eqn.exponentsrelation}).
The food webs are sorted according to their number of
edges}\label{fig.exponents}
\end{figure}

From Figure \ref{fig.exponents}, we know that most exponents of food
webs satisfy this relation, and that the errors become larger as the
scale of the network decreases. As the scaling behaviors we are
studying are statistical properties, the significance of these
regularities will increase with the number of samples. Therefore,
scaling behaviors are more obvious and accurate for large scale
networks because larger webs have more sample points.

\section{Discussion}
\label{sec.discussion}
\subsection{Ecological Meaning of Power Law Exponents}
As demonstrated above, food webs as energy transportation networks
always follow power law distributions and relations. Three important
exponents ($\alpha$,$\beta$ and $\eta$) are derived from these power
law regularities. The question of whether these exponents carry
ecological meaning naturally follows, and at first, the three
exponents all reflect integral properties of whole networks.
$\alpha$ describes the heterogeneity of first passage energy flux
distributions among vertices. The heterogeneity decreases with
$\alpha$. Therefore the distributions of energy flows are more
uneven in large food webs than the small ones. From table
\ref{tab.aidistribution}, we also know that all $\alpha$ values fall
into the interval $[1.23,1.77]$. According to the features of power
law distributions, the means and variances of power law random
variables with exponents smaller than 2 are
divergent\citep{Newman2005}. Therefore, energy flux on food webs has
no characteristic value. It is meaningless to find a specific
species with the average energy flux as the representation of other
species\citep{Newman2005}.

The allometric scaling relation describes the self-similarity of
flow networks. \cite{garlaschelli2003,allesina2005} pointed out that
allometric scaling exponents describe the transportation efficiency
of binary food webs because $C_i$ is treated as the cost of
transportation. The range of these exponents is between 1 (most
inefficient network) and 2 (most efficient network). However, we
believe that the exponent $\eta$ discussed in this paper does not
describe the efficiency of the whole network. As pointed out in
section \ref{sec.methods} and \ref{sec.example}, $C_i$ can be
understood as the energy store by the system, rooted from $i$ but
not the cost of the transportation. Thus, the food web with higher
$\eta$ can store more energy with the same consumption of
metabolites ($A_i$). Therefore, we believe that the food webs with
higher $\eta$ are more capable of storing energy by means of cycling
the flows in the network. In Table \ref{tab.allometricscaling}, we
see that the networks with larger scales have larger $\eta$ values.
Consequently, food webs can increase their ability to store energy
by increasing their complexity.

Further, the range of exponents $\eta$ is not simply $[1,2]$ (see
Table \ref{tab.allometricscaling}). As
\cite{garlaschelli2003,allesina2005,banavar1999} pointed out, the
range $[1,2]$ is only suitable for spanning trees or directed
acyclic graphs of the original binary food webs. However, our method
considers more ingredients, including the energy flux as the weight
of edges, the loop structures of energy flows, and the heterogenous
energy dissipation of each node, than the mere topology of the food
webs with homogenous nodes. That is the reason why the exponents are
out of the range $[1,2]$.

The exponent $\beta$ also describes the heterogeneity of indirect
effects. It is determined by exponents $\alpha$ and $\beta$ via the
theorem mentioned in section \ref{sec.exponentrelation}.

\subsection{Flux Matrices and Fundamental Matrices}
As discussed in section \ref{sec.methods}, $A_i$ and $C_i$ are
defined according to the flux matrix and fundamental matrix.
Therefore, the scaling behaviors of the food webs are determined by
the matrices. The properties of these matrices may help us to
understand the origin of the scaling behaviors.

The elements in flux matrices also follow power law distributions,
with an average exponent 1.46 (see Table
\ref{tab.matricesdistribution}). We hypothesis that this power law
determines the power law distribution of $A_i$.

\begin{table}
 \centering
 \caption{Parameters of the power law distributions for $f_{ij}$}
 {\small
($\alpha$ is the power law distribution exponent; $F_0$ is the
smallest value of $f_{ij}$ that follows the power law; $F_{max}$ is
the largest $f_{ij}$; $D$ is the KS statistic; $\sigma$ is the
quantile of 95\% confidence interval; number of samples is the total
number of edges following the power law distribution. The webs are
sorted by their number of edges)}

 \label{tab.matricesdistribution}
\begin{tabular}{lcccccc}
\hline Food web & $\alpha$ & $F_0/F_{max}$ & $D$ & $\sigma$ & No. of
Samples
\\ \hline
Baydry&1.33&3.82e-006&0.03&0.05&761\\
Baywet&1.33&4.91e-006&0.02&0.06&608\\
Florida&1.33&4.91e-006&0.02&0.06&608\\
Mangdry&1.32&2.81e-007&0.04&0.05&614\\
Everglades&1.37&2.98e-006&0.06&0.09&222\\
Gramdry&1.39&3.32e-006&0.06&0.10&194\\
Gramwet&1.37&2.98e-006&0.06&0.09&222\\
Cypdry&1.34&5.93e-005&0.05&0.10&203\\
Cypwet&1.29&1.15e-005&0.06&0.09&240\\
Mondego&1.40&4.80e-005&0.05&0.12&123\\
StMarks&1.84&1.29e-002&0.10&0.19&53\\
Michigan&1.43&5.88e-004&0.11&0.17&62\\
Narragan&1.70&1.43e-002&0.09&0.23&34\\
ChesUpper&1.32&1.95e-004&0.09&0.11&151\\
ChesMiddle&1.31&1.28e-004&0.11&0.12&134\\
Chesapeake&1.64&2.05e-002&0.17&0.24&30\\
ChesLower&1.65&1.17e-002&0.16&0.18&55\\
CrystalC&1.37&1.01e-004&0.08&0.24&30\\
CrystalD&1.29&6.56e-005&0.09&0.27&24\\
Maspalomas&1.71&2.10e-002&0.13&0.19&53\\
Rhode&1.84&2.65e-002&0.13&0.29&20\\

 \hline
\end{tabular}
\end{table}

However, unlike other variables, the distributions of fundamental
matrices are not power laws but rather more like log-normal because
the tails of the curves decline quickly (see Figure
\ref{fig.utilizationdistribution}). We also studied all of the
fundamental matrices of 21 empirical food webs, and noted that very
few could pass the KS test.

We presume that the calculation of the fundamental matrix in
equation \ref{eqn.utilizationmatrix} needs infinite operations on
$F$ matrix (see section \ref{sec.methods}). As a result, the noise
in $F$ matrices is enlarged and accumulated in the tails of the
distribution curves of $U$ matrix. However, the means by which the
non-power law distribution of fundamental matrices determines power
law distributions of $C_i$ and the allometric scaling power law
relationship is an interesting problem for future studies.

\begin{figure}
\centerline {\includegraphics{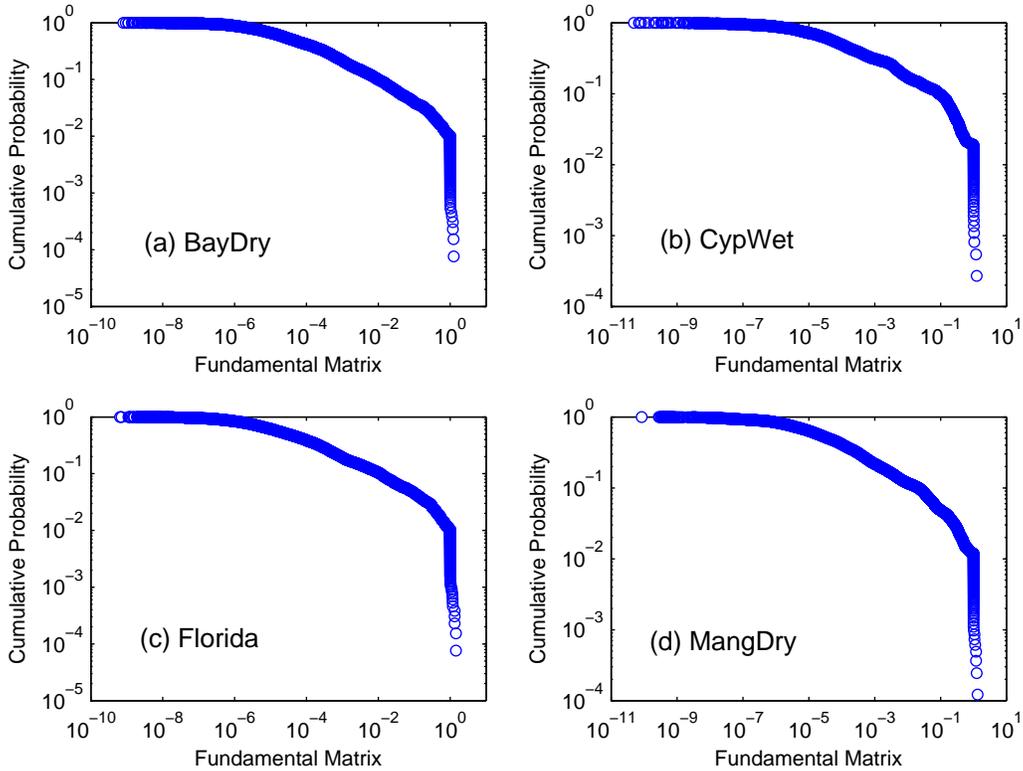}} \vskip3mm
\caption{Cumulative distribution of fundamental matrices of four
selected food webs}\label{fig.utilizationdistribution}
\end{figure}

\subsection{Information on Nodes}
One of the weak points of our study on allometric scaling power law
relationships is that the exponents are very close to 1. However, in
this case, "allometry" just means the non-linear relationship
between two variables, so, the relationship between $A_i$ and $C_i$
cannot be rigorously defined as an allometric scaling relationship.
Because the calculations of fundamental matrices and $C_i$s are
always based on linear algebra, the results are close to linear
relationship. One possible way to mend this weak point is to further
consider the information available about nodes.

Indeed, much information on nodes, i.e., each species in the web, is
ignored in this work. The biomass as the weight of each node is
available for many food webs. According to the definition, $C_i$ is
simply the energy store of the sub-system rooted from the vertex
$i$. Therefore, biomass information should be included in $C_i$
because a large part of energy will flow into the species node
stored as biomass. It is possible that a new approach of calculating
$C_i$ including the biomass information for all species may break
the linear relationship between $A_i$ and $C_i$.

Another important node characteristic is the body size of a given
species. The metabolic theory predicts that species body size of the
species can not only determine metabolism, life span, and birth
rate, etc. \citep{brown2004}, but also play an important role in
energy flows and food webs \citep{Cohen2003}. An integrated theory
of weighted food webs based on energy flow networks should contain
body size data.
\section{Concluding Remarks}
This paper presents a new approach to reveal the scaling natures of
weighted food webs as energy flow networks based on flux and
fundamental matrices. The $A_i$,$C_i$ distributions and the
relationship between them always follow power laws. The power law
exponents $\alpha$,$\beta$ and $\eta$ satisfy a relationship,
$\eta=(\alpha-1)/(\beta-1)$ as proved by the theorem
\ref{thm.exponent}. Power law exponents consistently change with
network scales.

We note that the allometric scaling exponent does not reflect the
transportation efficiency of networks but rather the capability of
storing energy, which is very different from previous studies. We
also investigated the distributions of flux matrices and fundamental
matrices, and suggested that biomass information should be
incorporated into future studies.

\paragraph{Acknowledgement}
Thanks for the support of National Natural Science Foundation of
China(No.70601002 and No.70771011). We acknowledge Clauset, A. for
providing the source code of power law fitting and KS test on his
web site; and also the Pajek web site to provide food web data
online. We also acknowledge three anonymous reviewers for advices.


\appendix

\section{A Simple Example}
\label{sec.example}

To understand the method introduced in the section
\ref{sec.methods}, let's look at a simple example (Figure
\ref{fig.example}).

\begin{figure}[htbp]
\centerline {\includegraphics{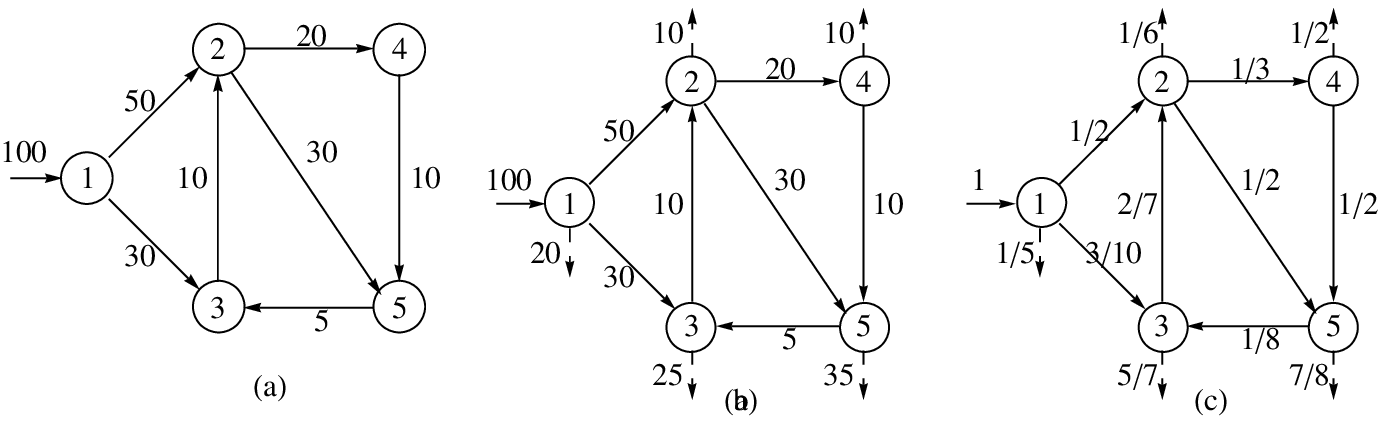}}\vskip3mm \caption{An
example to illustrate our method on deriving $A_i$s and $C_i$s
}{\small \textbf{(a)} is the original flow network; \textbf{(b)} is
the balanced network from (a), the dashed arrows are artificial
edges pointing to the sink node; \textbf{(c)} is the Markov chain
calculated according to (b).}\label{fig.example}
\end{figure}
As shown in Figure \ref{fig.example}, the balanced flow network and
the derived matrix $M$ can be obtained step by step according to the
method described in section \ref{sec.methods}. The fundamental
matrix $U$ can then be calculated for this simple network.
\begin{equation}
\label{eqn.exampleu}U = I+M+M^2+ \cdots=\left(
\begin{array}{ccccc}
 1 & 3/ 5 & 7/ 20 & 1/ 5 & 2/ 5\\
 0 & 42/ 41 & 7/82 & 14/41 & 28/41\\
 0 & 12/41 & 42/41 & 4/41 & 8/41\\
 0 & 3/164 & 21/328 & 165/164 & 21/41\\
 0 & 3/82 & 21/164 & 1/82 & 42/41
\end{array}
\right)
\end{equation}

Any entry $m_{ij}$ in the matrix $M$ is merely the probability of
one particle flowing from vertex $i$ to vertex
$j$\citep{Barber1978}. Furthermore, any entry $(i,j)$ in $M\cdot M$
represents the probability of a particle flowing from $i$ to $j$
along any path in $2$ steps. $M\cdot M\cdot M$ represents the
probabilities after $3$ steps, etc. \citep{higashi1993}. Thus, the
matrix $U$ simply takes in consideration all transfers of particles
along all possible paths.

Now, we will show how to calculate $A_i$ and $C_i$ for vertex 2.
According to the definition, $A_i$ is the total flow-through of
vertex $i$, so $A_2 = \sum_{j=1}^{6} {f_{2,j}}=60$.

Suppose many particles are flowing in the network. They will be
colored red once they flow through vertex $2$. These particles will
keep their color and flow around the whole network along all
possible pathways. Hence, the total number of red particles in the
network is just $C_2$, which is computed as,

\begin{equation}
\label{eqn.exampleC} C_2=G_2 \sum_{k=1}^5{ u_{2k}}=((100\times 3/5 +
0)/ (42/41))\sum_{k=1}^5{u_{2k}} =125
\end{equation}

where the term $G_2=\sum_{j=1}^5{f_{0j}{u_{j2}}/ {u_{22}}}$ in
equation \ref{eqn.exampleC} is the total number of new particles
which are colored red by vertex 2 in each time step.
$\sum_{j=1}^N{f_{0j}{u_{j2}}}$ is the number of particles that flow
into the system from the environment $0$ to the vertex $i$ along all
possible pathways at each time step, with $f_{0j}=(100,0,0,0,0)$ in
this example. By dividing by the term $u_{ii}$ to derive $G_i$ one
avoids double counting the red particles\citep{higashi1993}. Thus,
$C_i=G_i\sum_{k=1}^N{u_{ik}}$ is the total number of particles that
have been colored red and flow to other nodes along all possible
pathways at each time step.

If we treat the red particles flowing in the network as a metabolic
sub-system, we can calculate its allometric scaling relationship as
Garlaschelli has done for binary food webs \citep{garlaschelli2003}.
Thus, $A_i$ is the metabolism  and $C_i$ is the energy store or body
mass of the sub-system. Indeed, Garlaschelli's approach can be
recovered by our method, as shown in the next section.

\section{Comparisons to Existing Approaches}
\label{sec.comparison}

In this section, we will compare our approach to
\cite{garlaschelli2003}'s method and \cite{allesina2005}'s method.

\begin{figure}
\centerline {\includegraphics{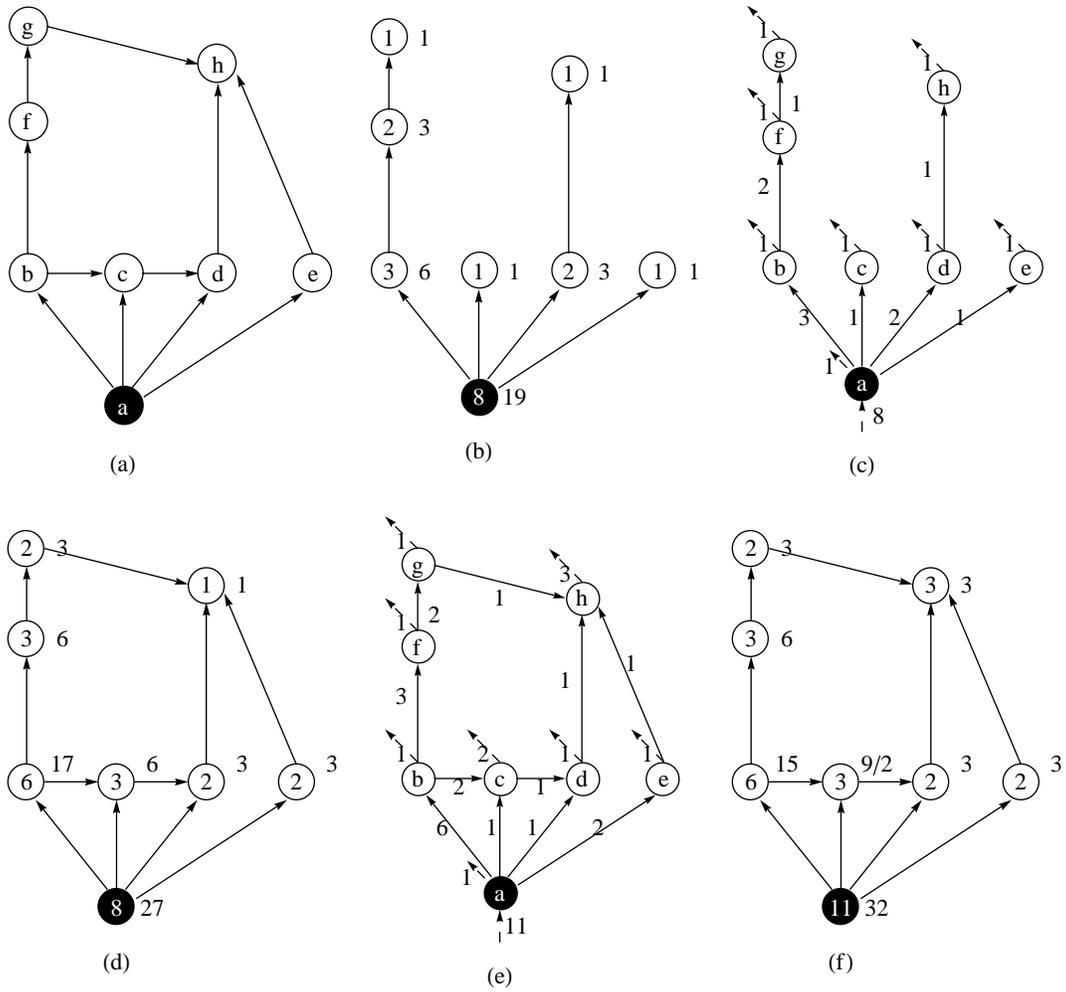}} \vskip3mm \caption{
Calculations of allometric scaling of a hypothetical food web by
various methods.}{\small {\textbf (a)} is a hypothetical food web
(The letter in each vertex is its index). The black vertex is the
root; {\textbf (b)} is a spanning tree of the original network (a).
$A_i$ and $C_i$ are denoted inside and beside vertex $i$; {\textbf
(c)} is the implicated flow network of (b), the numbers beside edges
are flux. The dashed lines are additional edges; {\textbf (d)} is a
directed acyclic graph of the original network, the numbers are
$A_i$s and $C_i$s calculated by the method of \citep{allesina2005};
{\textbf (e)} is a constructed flow network according to (d);
{\textbf (f)} is the network with numbers of $A_i$s and $C_i$s
calculated by our method according to the flow structure of
(e).}\label{fig.hfoodweb}
\end{figure}

Figure \ref{fig.hfoodweb}(a),(b) shows how Garlaschelli's approach
can be applied to a hypothetical food web to calculate $A_i$ and
$C_i$ for each vertex. At first, a spanning tree is constructed from
the original food web (Figure \ref{fig.hfoodweb} (a)) by cutting
edges. That way, each sub-tree rooted from any vertex can be viewed
as a sub-system of the spanning tree. For example, the sub-tree with
three vertices (b,f,g) rooted from the vertex b is a sub-system of
the spanning tree. $A_i$ is the total number of vertices involved in
this sub-tree and $C_i$ is the summation of $A_i$s for each vertex
in this sub-tree. Therefore, in this example, $A_b$ is 3 and $C_b$
is 6. Finally, the universal allometric scaling relationship of
$A_i$s and $C_i$s, with an exponent around $1.3$, was found for all
food webs, according to \cite{garlaschelli2003}.

Garlaschelli's method was inspired by \cite{banavar1999}'s model to
explain the Kleiber's law (See Figure \ref{fig.hfoodweb}(c)). The
spanning tree is simply Banavar's optimal transportation network.
Thus, energy flows into the whole system from the root along the
links of the network to all nodes. Suppose that each node would
consume 1 unit of energy in each time step. A flux with 1 unit
representing the energy consumption by each node should then be
added to the original spanning tree. In Figure
\ref{fig.hfoodweb}(c), the energy dissipation by each node is added
as a dotted line. As a result, $A_i$ of each node is just the total
influx of this node. $C_i$ is the total flux (the total number of
red particles colored by $i$) of the sub-tree rooted from $i$.
Essentially, calculation of allometric scalings using Garlaschelli's
approach is based on this weighted flow network model. Therefore,
our algorithm can derive the exact same values of $A_i$ and $C_i$
for the flow network (Figure \ref{fig.hfoodweb}(c)).

\cite{allesina2005} extended Garlaschelli's method. First, the
original network is converted to a directed acyclic graph (DAG) as
shown in Figure \ref{fig.hfoodweb}(d), then the sub-network
originated from vertex $i$ is identified as the set of vertices that
have at least one path from $i$. Therefore, vertices b,c,d,f,g,h
belong to the sub-network rooted from b because they are all
connected with b. $A_i$ is the number of vertices in the
sub-network, and $C_i$ is the summation of all $A_i$s in this
sub-network as shown in Figure \ref{fig.hfoodweb}(d). Because
Allesina and Bodini's approach is not based on weighted flows, it
cannot be covered by our approach.

However, we can construct a balanced flow network according to the
original network as shown in Figure \ref{fig.hfoodweb}(e). The
information of edges is added. Our approach can be applied to this
flow network to calculate $A_i$ and $C_i$ values (Figure
\ref{fig.hfoodweb}(f)). Comparing Figure \ref{fig.hfoodweb}(d) to
(f), we find their $A_i$s are almost the same, except vertices a and
h. To balance the network, more flows are added in vertex h and a,
so their $A_i$ values are larger than those in Figure
\ref{fig.hfoodweb}(d). This modification can not only influence
$A_i$ values but also $C_i$s. The vertices in networks
\ref{fig.hfoodweb}(d) and (f) therefore have different $C_i$ values.

Although our approach requires weight information, it can be
extended to more general flow networks, even those with cycles and
loops. Also, as demonstrated by our approach, $C_i$ merely means the
energy stored in the sub-system rooted from vertex $i$, which
provides much clearer and more significant ecological meaning than
previous works.

\section{A Theorem about Power Law Exponents}
\label{sec.theorem}

\begin{thm}\label{thm.exponent} Suppose $X$ and $Y$ are two random variables
following power law distributions. Their density functions,
$p(x)\thicksim x^{-\alpha}$ and $p(y)\thicksim y^{-\beta}$ hold for
any positive $x>x_0,y>y_0$, where $x_0$ and $y_0$ are lower bounds
of $X$ and $Y$. Additionally, $X$ and $Y$ satisfy a power law
relation, $Y\thicksim X^{\eta}$, then the exponents
$\alpha,\beta,\eta$ have following relationship:
\begin{equation}\label{eqn.exponentsrelation}
\eta={{(\alpha-1)}/{(\beta-1)}}
\end{equation}
\end{thm}

\begin{proof}
Because $X$ and $Y$ follow power law distributions,
$p(x)=cx^{-\alpha}$, where $c$ is a constant that satisfies the
normalization condition. And $Y= k X^{\eta}$, where $k$ is a
constant, so, for any $y>k x_0^{\eta}$,
\begin{equation}
\label{eqn.proof1}
P\{Y>y\}=P\{kX^{\eta}>y\}=P\{X>({y/k})^{1/\eta}\}=\int_{({y/k})^{1/\eta}}^{+\infty}cx^{-\alpha}dx
\end{equation}
Let $t=kx^\eta$, so $x=(t/k)^{1/\eta}$, $dx={(1/
\eta)}k^{-1/\eta}t^{(1/\eta)-1}dt$. Take it into equation
\ref{eqn.proof1},
\begin{equation}
\label{eqn.proof2} P\{Y>y\}=\int_{y}^{+\infty}(ck^{(\alpha-1)/\eta}/
\eta) t^{{(1-\alpha)/\eta} -1}dt=\int_{y}^{+\infty} k'
t^{{(1-\alpha)/ \eta}-1}dt
\end{equation}
where $k'=ck^{(\alpha-1)/\eta}/ \eta$ is a constant. Because $Y$
follows power law distribution,
\begin{equation}
\label{eqn.proof3} P\{Y>y\}=\int_{y}^{+\infty}c'y^{-\beta}dy
\end{equation}
Compare equation \ref{eqn.proof2} with equation \ref{eqn.proof3}, we
know that,
\begin{equation}
-\beta={(1-\alpha)/\eta} -1
\end{equation}
Finally,
\begin{equation}
\eta={(\alpha-1)/(\beta-1)}
\end{equation}
\end{proof}

\bibliographystyle{elsarticle-harv}
\bibliography{ecology}

\end{document}